\documentclass[
preprint,
12pt,
tightenlines,
showpacs,
]{revtex4-1}
\pdfoutput=1
\usepackage[T1]{fontenc}
\usepackage[latin9]{inputenc}
\usepackage{graphicx}
\usepackage{setspace}
\usepackage{amsmath}
\usepackage{amsthm}
\usepackage{enumerate}
\usepackage{textcomp}
\usepackage{amssymb}
\usepackage{relsize}
\usepackage{appendix}
\usepackage{babel}
\usepackage{amsfonts}
\usepackage{tabulary}
\usepackage{appendix}

\usepackage{rotating}
\usepackage{enumitem}
\usepackage{titlesec}
\usepackage{bm}

\titlespacing\subsection{0pt}{12pt plus 4pt minus 2pt}{1pt plus 2pt minus 2pt}


\newtheorem{lem}{Lemma}
\begin{document}

\title{Controlled Perturbation-Induced Switching \linebreak in Pulse-Coupled Oscillator Networks }

\author{Fabio Schittler Neves}
\affiliation{Network Dynamics Group, Max Planck Institute for Dynamics and Self-Organization, G{\"o}ttingen, D-37073, Germany.}
\affiliation{Bernstein Center for Computational Neuroscience (BCCN), G{\"o}ttingen, Germany.}
\email{schittler@gmail.com}

\author{Marc Timme}
\affiliation{Network Dynamics Group, Max Planck Institute for Dynamics and Self-Organization, G{\"o}ttingen, D-37073, Germany.}
\affiliation{Bernstein Center for Computational Neuroscience (BCCN), G{\"o}ttingen, Germany.}
\email{timme@nld.ds.mpg.de}

\begin{abstract}
Pulse-coupled systems such as spiking neural networks exhibit nontrivial invariant sets in the form of attracting yet unstable saddle periodic orbits where units are synchronized into groups. Heteroclinic connections between such orbits may in principle support switching processes in those networks and enable novel kinds of neural computations. For small networks of coupled oscillators we here investigate under which conditions and how system symmetry enforces or forbids certain switching transitions that may be induced by perturbations. For networks of five oscillators we derive explicit transition rules that for two cluster symmetries deviate from those known from oscillators coupled continuously in time. A third symmetry yields heteroclinic networks that consist of sets of all unstable attractors with that symmetry and the connections between them. Our results indicate that pulse-coupled systems can reliably generate well-defined sets of complex spatiotemporal patterns that conform to specific transition rules. We briefly discuss possible implications for computation with spiking neural systems.
\end{abstract}
\maketitle

\section*{Introduction}
Heteroclinic connections among saddle states are known to support non-trivial switching dynamics in networks of units coupled continuously in time \cite{GH1988,H1993,R2001,PD2005,DT2006}. Interesting recent work furthermore suggests that heteroclinic networks in state space may be used to encode a large number of spatiotemporal patterns if the transition between different states is controllable \cite{AB2004}. Supplementing such systems with certain additional features may thus enable a new kind of computation \cite{AB2005}.

Networks of pulse-coupled oscillators, that model, e.g., the dynamics of spiking neural networks, constitute hybrid systems that are very distinct from systems coupled continuously in time. In pulse-coupled hybrid systems pulses interrupt the otherwise smooth time evolution at discrete event times when pulses are sent or received. Such networks may exhibit unstable attractors \cite{TWG2002}, unstable saddle periodic orbits that are attractors in the sense of Milnor \cite{K1997}. Recent works indicate that unstable attractors may generically occur in systems with symmetry \cite{AT2005,BES2008-1} and that such saddle periodic orbits may be connected to heteroclinic networks in a standard way, but with some anomalous features \cite{BES2008-2,CT2008}. In particular, due to the attractor nature of the periodic orbits, switching among saddles requires external perturbations. It was known before that in non-hybrid systems with non-attracting saddles, such perturbations may in principle direct the switching path. In this work we study under which conditions and precisely how small controlled perturbations can exploit heteroclinic connections in pulse-coupled systems to support switching processes among saddle states, a key prerequisite for computation by heteroclinic switching. 

The results may be of particular relevance for neural systems where pulses are electric action potentials (spikes) generated by neurons because spatiotemporal switching patterns of spikes have been suggested to underly information processing \cite{AT2005-2,R2006}.

This article is divided into four main sections. After introducing the model and explaining our analytical approach in the first section, in the second we present the most persistent attractors and their symmetries. In the third section, we derive the dynamic response of the system to single oscillator perturbations and provide a local stability analysis. Finally, we conclude discussing the relation between switching processes in the pulse-coupled systems considered to those in systems coupled continuously in time. We also briefly discuss potential implications for neural coding and paths to future investigations.

\section{Pulse-coupled network}\label{sDescription}
Consider a network of $N$ oscillators that are connected homogeneously all-to-all without self-connections through delayed pulse-couplings. The state of each oscillator $i\in \{1,\ldots,N \}$ at time $t$ is specified by a single phase-like variable $\phi_{i}(t)$ \cite{MS1990}. In the absence of interactions, its dynamics is given by
\begin{equation}\label{eq1}
\frac{d\phi_{i}}{dt}=1, \qquad 0 \leq  \phi_{i} \leq 1.
\end{equation}
When oscillator $i$ reaches a threshold, $\phi_{i}(t^-)=1$, its phase is reset to zero, $\phi_{i}(t)=0$, and the oscillator is said to send a pulse. Such pulse is sent to all other oscillators which receive this signal after a delay time $\tau$. The incoming signal induces a phase jump
\begin{equation}\label{eq2}
\phi_{i}\left(t\right)=H_{\epsilon}(\phi_{i}(t^{-}))=U^{-1}\left[U\left(\phi_i(t^{-})\right)+\epsilon\right], 
\end{equation}
which depends on the instantaneous phase $\phi_{i}(t^-)$ of the post-synaptic oscillator and the excitatory coupling strength $\epsilon>0$. The phase dependence is determined by a twice continuously differentiable potential function $U(\phi)$ that is assumed to be strictly increasing ($U^{'}(\phi) > 0$), concave down ($U^{''}(\phi)< 0$), and normalized such that $U(0)=0$, $U(1)=1$. As shown in \cite{TWG2002, TWG2003}, this phase dynamics is equivalent to the ordinary differential equations
\begin{equation}\label{neweq2}
\frac{dV_{i}}{dt^{'}}=f(V_{i})+S_{i}(t^{'}),
\end{equation}
where
\begin{equation}
S_{i}(t^{'})=\sum_{ j=1 \atop j\neq i }^{N}\sum_{k \in \mathbb{Z}} \epsilon \delta \left(  t -\tau^{'} -t^{'}_{jk} \right),
\end{equation}
is a sum of delayed $\delta$-currents induced by presynaptic oscillators.  Oscillator $j$ sends its $k$th pulse at time $t^{'}_{jk}$ whenever its phase variable crosses threshold, $V_{j}(t^{'-}_{jk}) \geq 1$; thereafter, it is instantaneously reset, $V_{j}(t^{'}_{jk})\rightarrow 0$.  The $k$th pulse of oscillator $j$ is received by $i$ after a delay $\tau^{'}$. The positive function $f(V) > 0$ yields a free ($S_{i}(t^{'})\equiv 0$) solution $V_i(t^{'}):=V(t^{'})=V(t^{'}+T_{0})$ of intrinsic period $T_{0}$. The above function $U(\phi)$ is related to this solution via
\begin{equation}\label{neweq3}
U\left(\phi\right)=V\left(\phi T_{0}\right),
\end{equation}
defining a natural phase $\phi$ by rescaling the time axis, $t=t^{'}/T_{0}$ and $\tau=\tau^{'}/T_{0}.$

We focus on the specific form $U\left(\phi_{i} \right)=U_{IF}\left(\phi_{i} \right)=\frac{I_{i}}{\gamma}\left(  1-e^{-\phi_{i} T_{IF} } \right)$ that represents the integrate-and-fire oscillator defined by $f(V)=I-\gamma V$. Here $I>1$ is a constant external input and $T_{IF}=\frac{1}{\gamma}\log{\left(1-\frac{\gamma}{I}\right)}^{-1}$ the intrinsic period of an oscillator. Any $U(\phi)$ sufficiently close to $U_{IF}(\phi)$ give qualitatively similar results.

After defining the dynamics of the network elements we can define its collective state as a phase vector,
\begin{equation}\label{eq5}
\bm{\phi} = \left(\phi_{1},\phi_{2},\ldots,\phi_{N}\right),
\end{equation}
where each $\phi_{i}$ describes the phase of oscillator $i$. Their dynamics is governed by (\ref{eq1}) and (\ref{eq2}). The difference in phase among them will define the macroscopic states of the network, as explained in the next section.

The event-based updating presented above brings two main advantages: It yields exact analytical solutions of state space trajectories and substantially reduces the simulation time compared to numerical integration with fixed small time steps.

\section{Periodic orbit dynamic and symmetries}\label{sec-PDY}
Here we define and explicitly study the dynamics of partially synchronized states, periodic orbits where groups of oscillators are identically synchronized into clusters, for three main symmetries of $N=5$ oscillators. The analysis reveals mechanisms of perturbation-induced switching transitions that critically depend on the local stability properties of cluster periodic orbits. As we show below, stability in turn is determined by whether a cluster receives only sub-threshold input during one period (``unstable'' cluster) or it also receives supra-threshold input (``stable'' cluster), the only two options available. Thus, similar switching mechanisms for a given symmetry  will prevail also for larger $N$, cf \cite{TWG2003}, and contribute to much more complex saddle state transitions, cf figure \ref{diagram-s2s2s1}.
As shown in the last section, when a constant external input $I$ to a single oscillator $i$ is sufficiently strong to drive the membrane potential to cross its threshold ($U^{'}>0$), the potential dynamics becomes periodic with period $T_{0}$. It was known before that networks of such pulse-coupled oscillators may exhibit different invariant states including partially synchronized states \cite{DT2006,K1997,UKT1995,TWG2002,KGT2009}.

To explore the possible unstable attractors we systematically varied the parameters and the initial conditions for our system and found numerically that three clustered states present those state symmetries most persistent to perturbations. Two of these states are composed of two clusters, with permutation symmetries $S_{3} \times S_{2}$ and $S_{4} \times S_{1}$, respectively; another one is composed of two clusters and one single element, with permutation symmetry $S_{2} \times S_{2} \times S_{1}$. The event-based analyses of these states is based in return maps that are presented in detail in tables \ref{ut-s3s2-1}, \ref{ut-s3s2-2}, \ref{ut-s4s1} and \ref{ut-s2s2s1}. For each of these three cluster periodic orbits, the event sequence of sending and reception of pulses fully defines the type of periodic orbit such that the analytical conditions for existence of a family of such orbits can be directly read from these tables. In particular, these three families of periodic orbits exist for an open set of parameters  close to the three examples numerically specified in tables \ref{dt-s3s2}, \ref{dt-s4s1} and \ref{dt-s2s2s1}. The existence conditions for each periodic orbit naturally imply that the phases of all oscillators exactly return to the same value after a fixed period; at the same time, the predefined event sequence must be kept. 

Throughout this work, we represent the dynamical states relative to the symmetries $S_{3} \times S_{2}$, $S_{4} \times S_{1}$ and $S_{2} \times S_{2} \times S_{1}$, respectively, by the phase vectors
\begin{subequations}
\begin{equation}
\bm{\phi} = \left(a,a,a,b,b\right),
\end{equation}    
\begin{equation}
\bm{\phi} = \left(a,a,a,a,b\right),
\end{equation}
\begin{equation}
\bm{\phi} = \left(a,a,b,b,c\right),
\end{equation}
\end{subequations}
where each element represents one oscillator and the letters indicate to which cluster it belongs. In these periodic orbits the differences in phase $\left(a-b\right)$, $\left(a-c\right)$ and $\left(b-c\right)$ will change in time in a periodic manner, while the cluster configuration remains the same. In this notation the elements labeled as 'a' belong to an unstable cluster, as 'b' to a stable cluster, while 'c' represents a single element that reacts stably to small perturbations (see below).

It is important to emphasize that this system exhibits symmetric connections and that the parameters, $I$, $\gamma$, $\epsilon$ and $\tau$, are global (see section \ref{sDescription}).  As a consequence, the initial condition controls the final attractor and determines which of the permutation-equivalent states is obtained. By the same symmetry argument, the number of permutation-equivalent configurations for each state symmetry is given by the number of ways we can form the vectors presented above, which results in 10, 5, and 30 states, respectively. In the next section we study the stability of these cluster states and the possible state transitions among them in the presence of small perturbations.

\section{Stability and switching properties}\label{switching}
In this section we will study, case by case, the dynamics and stability of cluster periodic orbits presented in the last section. First we show that these periodic orbits actually are unstable attractors, and later we study the possible transitions between different states in response to small perturbations.

To study the local stability of these attractors, we introduce a perturbation vector,
\begin{equation}
\bm{\delta}(n) = \left(\delta_{2}(n),\delta_{3}(n),\delta_{4}(n),\delta_{5}(n)\right),
\end{equation}
that has four components since only the relative phases among the oscillators are relevant ($\delta_{1}(n) \equiv 0$). The analysis presented here consists of a study of the temporal evolution of this perturbation vector at each cycle. Thus, 
\begin{equation}
\delta_{i}(n) := \phi_{i}(t_{1,N})-\phi_{i}^{*}(t)
\end{equation}
are the perturbations to phases on the periodic orbit just after oscillator one has sent its $n$th pulse and been reset, i.e. $\delta_{i} \equiv 0$.

After a small enough initial perturbation that is added to the phase vector at some point of the unperturbed dynamic, the temporal evolution of the perturbation vector is defined as the difference between this perturbed vector after one cycle of the system dynamic and the unperturbed phase vector at the same time. Analytically tracking the periodic orbit dynamics (cf tables \ref{ut-s3s2-1}, \ref{ut-s4s1}, and \ref{ut-s2s2s1}) yields the perturbation perturbation vector after one cycle as a function of the perturbation in the previous cycle, 
\begin{equation}
\bm{\delta}(n+1) = \bm{F}\left(\bm{\delta}(n)\right),
\end{equation}
which can be linearly approximated by
\begin{equation}\label{linear-aprox}
\bm{\delta}(n+1) \doteq  \bm{J}\bm{\delta}(n),
\end{equation}
where $\bm{J}$ is the Jacobian matrix at $\bm{\delta}(n)=0$, describing the local dynamics.

After analyzing the local stability properties, we study non-local effects in response to single oscillator perturbations. The procedure consists of perturbing only one oscillator at each time. We consider negative perturbations, instantaneous decrements on the phase, and positive ones, instantaneous increments on the phase. When possible, transition diagrams are included.
\begin{figure}[]
\begin{center}
\includegraphics[width=11.0cm,clip,angle=0]{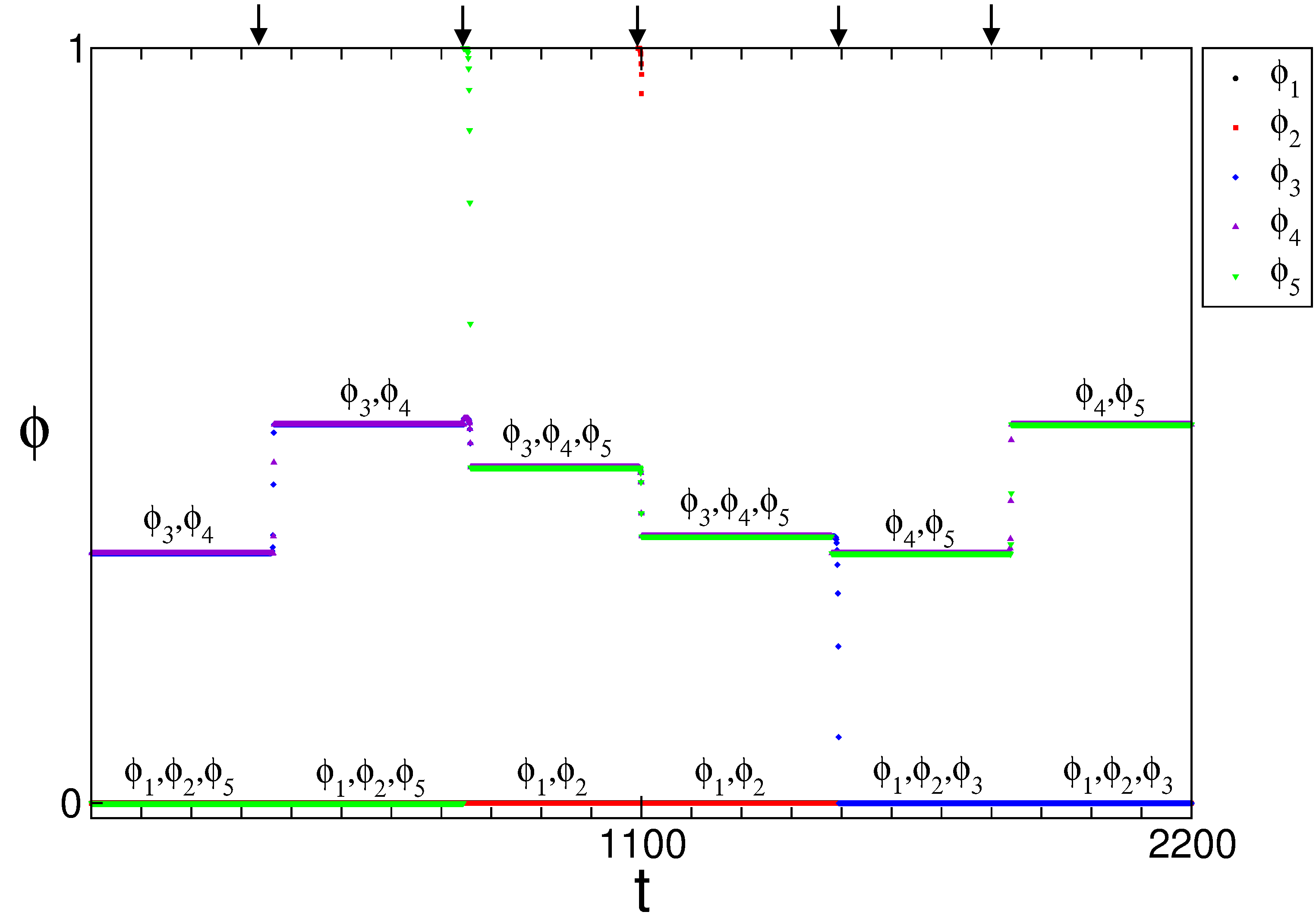}
\end{center}
\caption[1st-entry]{  \footnotesize Example of perturbation-induced switching in a $S_{3} \times S_{2}$ state set. The response of the system to a sequence of five negative single oscillator perturbations preserving a $S_{3} \times S_{2}$ clustered symmetry ($\tau=0.31$, $\epsilon=0.025$, $I=1.04$, $\gamma=1$). The phases of all oscillators are plotted at the moment when oscillator 1 is reset, each color representing the phase of one oscillator. There are transitions through two steps, where in a first moment the cluster $S_{3}$ is unstable, after one perturbation (as shown by the second and fourth perturbation) it reaches a new configuration with the cluster $S_{2}$ in a unstable phase position, a second perturbation (as shown in the first, third and fifth perturbations) is needed to put the system in the initial phase difference again, maintaining the cluster components but changing its stability. The symmetry of the unstable attractors are preserved. The sequence of states given by the plateaus are $(a,a,b,b,a)^{*} \rightarrow (a,a,b,b,a)\rightarrow (b,b,a,a,a)^{*} \rightarrow (b,b,a,a,a) \rightarrow (a,a,a,b,b)^{*} \rightarrow (a,a,a,b,b)$, where the star indicate the states where $S2$ is in a unstable phase.}
\label{fig3}
\end{figure}
\subsection{Clustered state $S_{3} \times S_{2}$}
For the clustered state $S_{3} \times S_{2}$ we have the temporal evolution for the vector $\bm{\delta}$ in one cycle, (cf tables \ref{pt-s3s2} and \ref{ut-s3s2-1}) assuming $\delta_{2}<\delta_{3}$, and $\delta_{4}<\delta_{5}$, given by:
\begin{equation}\label{delta-s3s2}
\bm{\delta}(n+1) = \bm{\phi} - \left(0,0,A,A\right)
\end{equation}
where $\bm{\phi}$ is the phase vector given by the last row of table \ref{pt-s3s2} and the vector $\left(0,0,A,A\right)$ represents the unperturbed cycle (see table \ref{ut-s3s2-1}). By (\ref{linear-aprox}) and (\ref{delta-s3s2}), we obtain the following Jacobian matrix:
\begin{equation} \label{J-s3s2}
{\frac{\partial\bm{\delta}(n+1)}{\partial  \bm{\delta}(n)} \bigg{|}}_{\bm{\delta}(n)=0} = \left( 
\begin{array}{cccc}
\alpha & 0 & 0 & 0 \\
j_{21} & \beta & 0 & 0 \\
\gamma & j_{32} & 0 & 0 \\ 
\gamma & j_{42} & 0 & 0 \\
\end{array} \right) ,
\end{equation}
where $\alpha$, $\beta$, and $\gamma$ are positive reals larger than one, $j_{32}$ and $j_{42}$ are positive, and $j_{21}$ is much smaller than the other elements (for analytical expressions of the partial derivatives refer to \ref{app-deriv}).
This matrix has two zero eigenvalues with eigenvectors that correspond to the directions of  $\delta_{4}$ and $\delta_{5}$; and two non-zero eigenvalues given by,
\begin{eqnarray}
\lambda_{1}=\alpha=
\left[-1+\left[1+H_{\epsilon}^{'}(\tau)\right]H_{\epsilon}^{'}(H_{\epsilon}(\tau))\right]H_{2\epsilon}^{'}(\tau+H_{2\epsilon}(\tau))\\
\lambda_{2} =\beta= \left[-1+2H_{\epsilon}^{'}(\tau)\right]H_{\epsilon}^{'}(H_{\epsilon}(\tau))H_{2\epsilon}^{'}(\tau+H_{2\epsilon}(\tau)),
\end{eqnarray}
where
\begin{equation}
H_{\epsilon}^{'}(\phi)=\frac{\partial}{\partial \phi}U^{-1}\left(U(\phi)+\epsilon\right).
\end{equation}

Here $\lambda_{1}$ and $\lambda_{2}$ are larger than one (see Lemma \ref{lemma1}), noticing that all terms are due to sub-threshold events. Thus a perturbation can effectively disturb the system in two different possible directions, showing that the cluster $S_{2}$ is stable and the $S_{3}$ is unstable.

\begin{lem}\label{lemma1}
If $H_{\epsilon}(\phi)$ given by (\ref{eq2}) mediates a sub-threshold reception event and $\epsilon>0$, $U^{'}(\phi)>0$, and $U^{''}(\phi)<0$,  then $H_{\epsilon}^{'}(\phi)>1$.
\end{lem}
\begin{proof}
Assume $\epsilon>0$. By definition 
\begin{equation*}
H_{\epsilon}^{'}(\phi)=\frac{\partial}{\partial \phi}U^{-1}\left(U(\phi)+\epsilon\right) = \frac{U^{'}(\phi)}{U^{'}\left(U^{-1}(U(\phi)+\epsilon)\right)} = \frac{U^{'}(\phi)}{U^{'}\left(H_{\epsilon}(\phi)\right)},
\end{equation*}
Since $U^{'}(\phi)$ is a monotonic decreasing function and $H_{\epsilon}(\phi)>\phi$ we have $U^{'}(\phi)>U^{'}(H_{\epsilon}(\phi))$ for any $H_{\epsilon}(\phi)$, and consequently $H_{\epsilon}^{'}(\phi)>1$.
\end{proof}

Now we describe the long-term effect of a single oscillator perturbation to the unstable cluster $S_{3}$. A negative perturbation to one of the elements on the unstable cluster ($ \phi^{+} = \left(a,a,a-\delta_{3},b,b\right)$) puts one of its elements phase slightly behind; then the initial stable cluster $S_{2}$ begins to receive an additional pulse just after it is reset, increasing its relative phase in each cycle, and thus approaching the phase of the elements in the originally stable cluster. After some cycles it finally joins that cluster by a simultaneous reset, forming a new $S_{3} \times S_{2}$ clustered state. This switching process is illustrated just after the second and fourth perturbations in figure \ref{fig3}. The final state has the same symmetry as the initial state, but has different stability properties: Whereas the orbit is stable to splitting the $S_3$ cluster, it is unstable to splitting the $S_2$ cluster and upon perturbation resynchronizes and shifts in phase with respect to the cluster $S_3$ (see table \ref{ut-s3s2-2}). A further perturbation to the cluster $S_{2}$ does not change the elements of each cluster but just returns the system to the initial phase difference, as illustrated in the first, third and fifth perturbations in figure \ref{fig3}. 
\begin{figure}[b]
\begin{center}
\includegraphics[width=11.0cm,clip,angle=0]{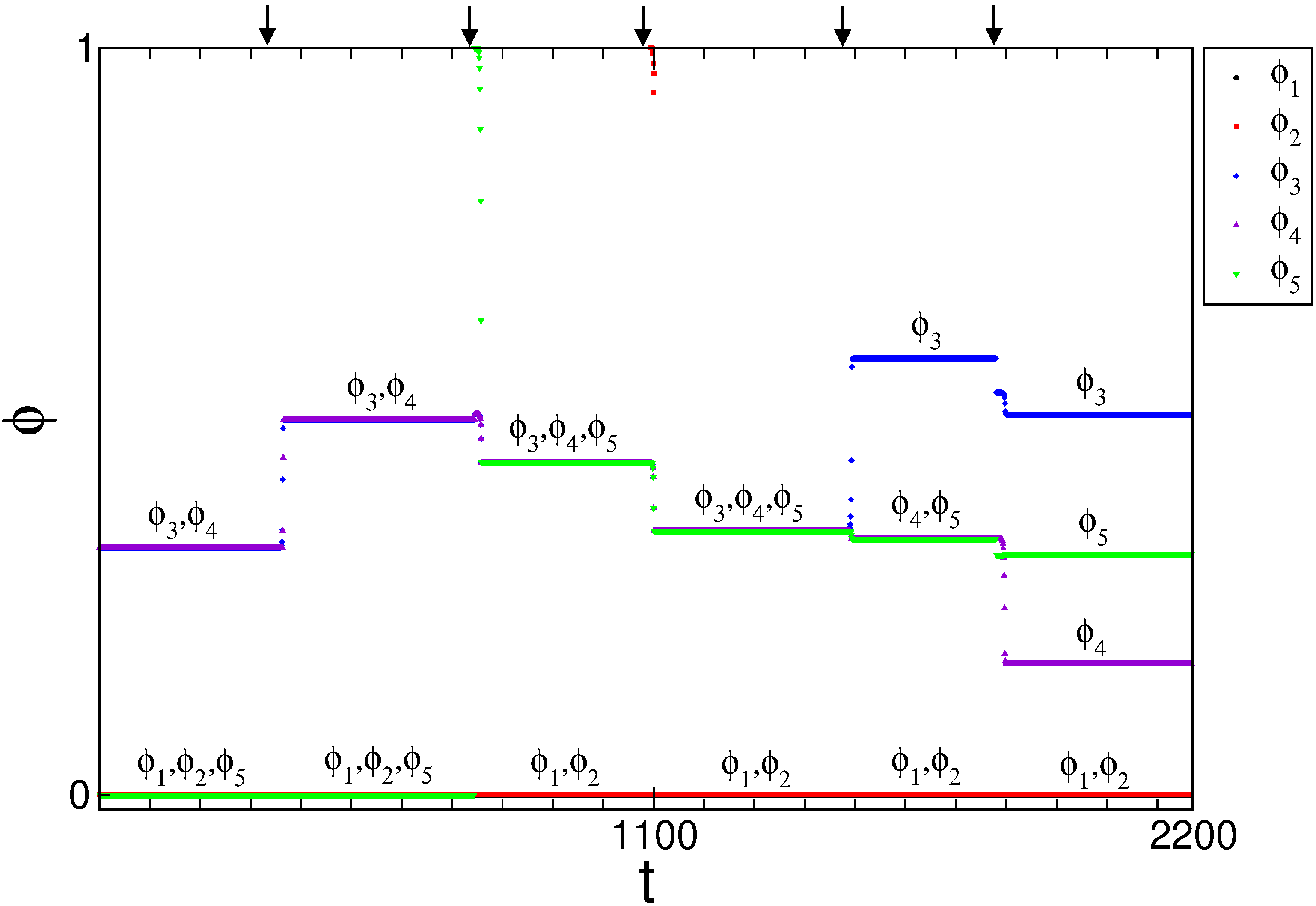}
\end{center}
\caption[1st-entry]{ \footnotesize Example of symmetry change by perturbation-induced switching. A sequence of three negative and two positive single oscillator perturbations (same parameters as in figure \ref{fig3}). Showing that the positive perturbations splits the unstable cluster until the system reach some periodic orbit. The symmetry is not preserved.}
\label{fig4}
\end{figure}

Intriguingly, positive perturbations $\bm{\phi}^{+} = \left(a,a,a+\delta_{3},b,b\right)$) result in a completely different dynamic, as can be seen in figure \ref{fig4} which presents a sequence of three negative and two positive perturbations. A positive perturbation puts just one oscillator from the unstable cluster ahead, that now increases its phase in relation to its original cluster in each cycle till it begins to be reset by pulses coming from the originally stable cluster. The original $S_{2}$ cluster changes its phase to conform with this new pulse configuration, but still been reset by pulses coming from the two elements left on the unstable cluster. Thus the $S_{3}$ cluster splits into two clusters, and the new configuration becomes $S_{2} \times S_{2} \times S_{1}$. A further perturbation puts the system in a stable cyclic state.

Hence the symmetry $S_{3} \times S_{2}$ is not preserved upon a general perturbation. however, simulations suggest that if we only consider negative single oscillator perturbations, or one negative perturbation with a larger magnitude than the others, the symmetry is preserved, and it is possible to write a transition rule,
\begin{equation}
 (a-\delta_{1},a,a,b,b) \rightarrow (a,b,b,a,a),
\end{equation}
that permits us to know which will be the next state after one perturbation

\subsection{Clustered state $S_{4} \times S_{1}$}
Considering now the $S_{4} \times S_{1}$ symmetry, assuming $\delta_{2}<\delta_{3}<\delta_{4}$ and $\delta_{5} > 0$, in an procedure analogous to that in the last section, we obtain the following Jacobian matrix (see last row of tables \ref{pt-s4s1} and \ref{ut-s4s1}):
\begin{equation} \label{J-s4s1}
{\frac{\partial\bm{\delta}(n+1)}{\partial\bm{\delta}(n)} \bigg{|}}_{\bm{\delta}(n)=0} = \left( \begin{array}{cccc}
\alpha & 0 & 0 & 0 \\
j_{21} & \beta & 0 & 0 \\
j_{31} & j_{32} & \gamma & 0 \\ 
j_{41} & j_{42} & \theta & 0 \\
\end{array} \right) ,
\end{equation}
here $\alpha$, $\beta$, $\gamma$, and $\theta$ are larger than one, $j_{41}$, $j_{42}$ are positive, and $j_{21}$, $j_{31}$, and $j_{32}$ are much smaller than the other elements (see \ref{app-deriv}). This matrix has one zero eigenvalue, corresponding to the single element represented by $S_{1}$ and three non-zero eigenvalues given by
\begin{flalign}
\lambda_{1}&=\alpha =\left[-1+\left[1+H_{\epsilon}^{'}(\tau)H_{\epsilon}^{'}\left(H_{\epsilon}(\tau)\right) \right]H_{\epsilon}^{'}\left(H_{2\epsilon}(\tau)\right)\right]H_{\epsilon}^{'}\left(\tau+H_{3\epsilon}(\tau)\right), \\
\lambda_{2}&= \beta=\left[-1+\left[1+H_{\epsilon}^{'}(\tau)\right]H_{\epsilon}^{'}\left(H_{\epsilon}(\tau)\right) \right]H_{\epsilon}^{'}\left(H_{2\epsilon}(\tau)\right)H_{\epsilon}^{'}\left(\tau+H_{3\epsilon}(\tau)\right), \\
\lambda_{3}&=\gamma=\left[-1+2H_{\epsilon}^{'}(\tau)\right]H_{\epsilon}^{'}\left(H_{\epsilon}(\tau)\right)H_{\epsilon}^{'}\left(H_{2\epsilon}(\tau)\right)H_{\epsilon}^{'}\left(\tau+H_{3\epsilon}(\tau)\right).
\end{flalign}
Using the same argument as in the last section (see lemma \ref{lemma1}), these three eigenvalues are necessarily larger than one ($\lambda_{1},\lambda_{2},\lambda_{3}>1$), showing the instability of the cluster $S_{4}$.

Again we study the effect of single oscillator perturbations. As can be seen in figure \ref{fig5}, positive perturbations to an element of the unstable cluster $S_{4}$ put one element from the unstable group ahead. This difference in phase will increase each cycle, since pulses received at larger phases more strongly shift the oscillator's phase; at the same time the $S_{1}$ element is  not only reset by the other oscillators pulses, but also receives additional pulses that makes its phase approach the unstable cluster. After some cycles, the element $S_{1}$ joins the original unstable cluster, forming a new $S_{4}$ cluster while the perturbed oscillator forms the new $S_{1}$, returning the system to its original phase difference and symmetry, but with different elements composing the clusters.

\begin{figure}[]
\begin{center}
\includegraphics[width=11.0cm,clip,angle=0]{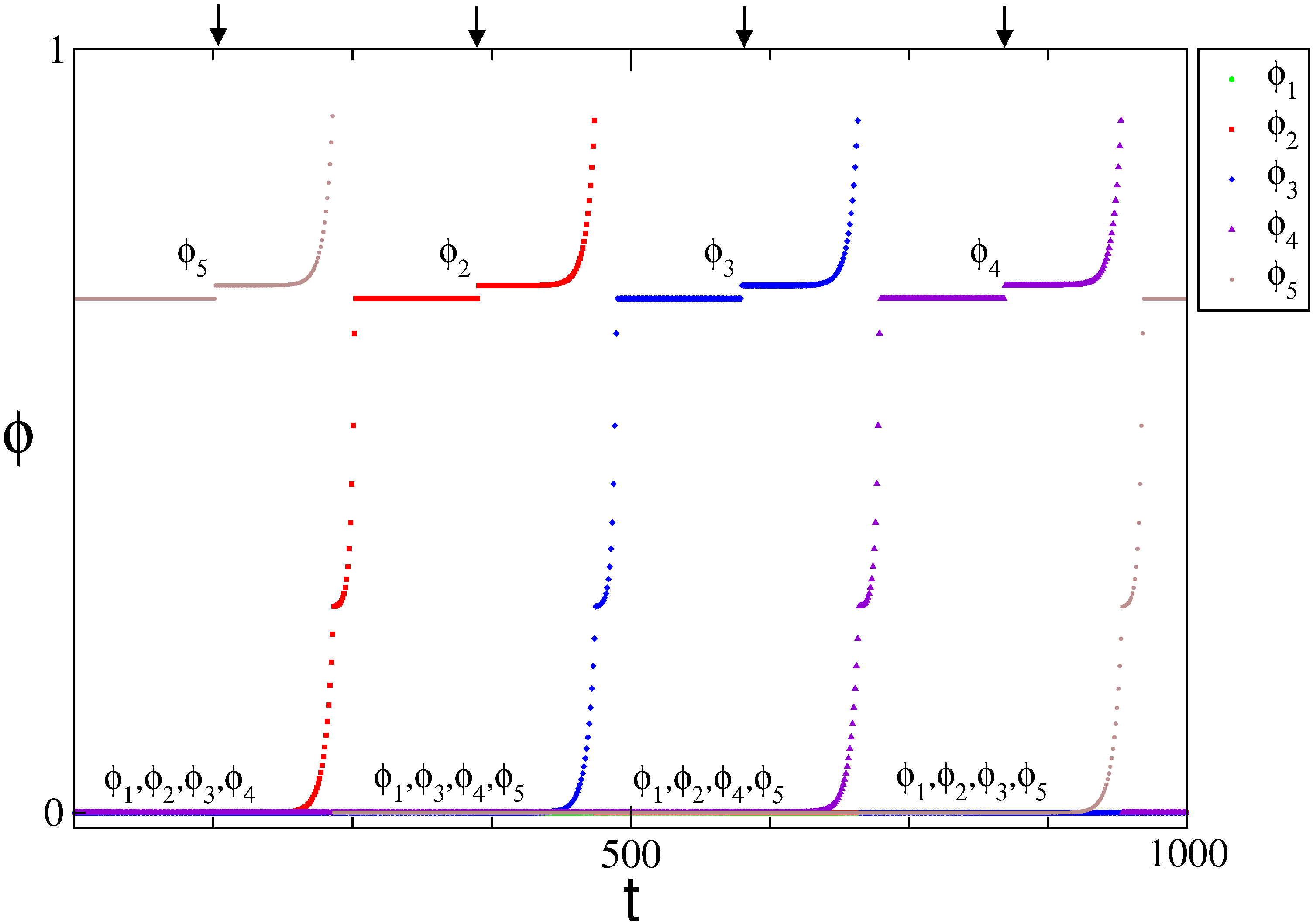}
\end{center}
\caption[1st-entry]{ \footnotesize Example of perturbation-induced switching in a $S_{4} \times S_{1}$ state set. A sequence of four positive single oscillator perturbations preserving a $S_{4} \times S_{1}$ clustered state ($\tau=0.27$, $\epsilon=0.015$, $I=1.1$, $\gamma=1$). The phase of all oscillators are plotted each time oscillator 1 is reset, each color represents the phase of one oscillator. The perturbed oscillator leaves the cluster $S_{4}$ and replaced the $S_{1}$ oscillator, that join the cluster $S_{4}$, preserving the symmetry. The sequence of states corresponding to the plateaus are $(a,a,a,a,b) \rightarrow (a,b,a,a,a)\rightarrow (a,a,b,a,a) \rightarrow (a,a,a,b,a)$.}
\label{fig5}
\end{figure}
 
\begin{figure}[]
\begin{center}
\includegraphics[width=11.0cm,clip,angle=0]{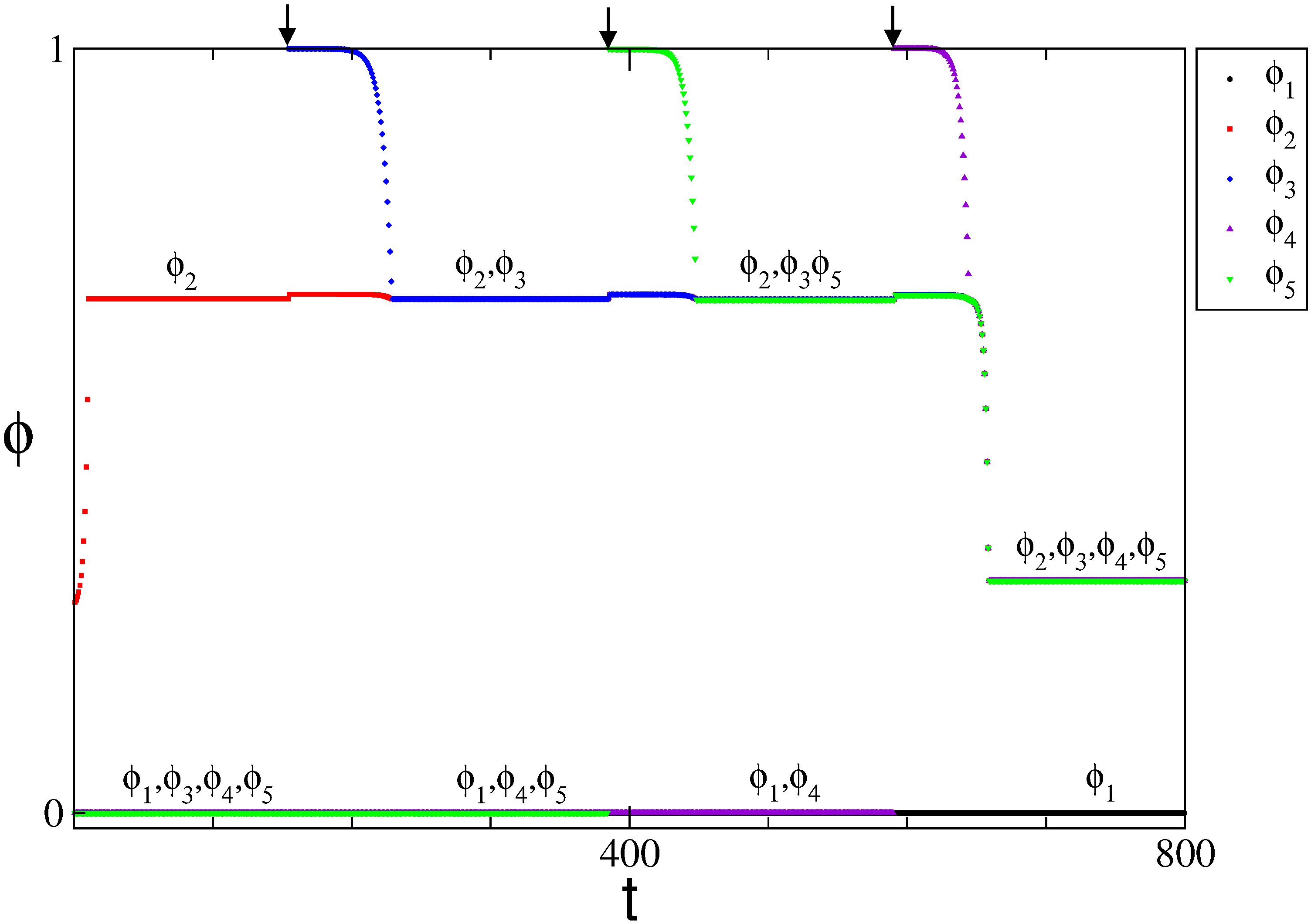}
\end{center}
\caption[1st-entry]{ \footnotesize  Example of symmetry change by perturbation-induced switching. A sequence of three negative single oscillator perturbations for the same parameter used on figure \ref{fig5}. The perturbed oscillator joins the $S_{1}$ oscillator, forming a $S_{2}$ cluster, new perturbations led first to an $S_{3} \times S_{2}$ configuration and later to split the clusters reaching a stable attractor. The symmetry is not preserved.}
\label{fig6}
\end{figure}

When a negative perturbation is applied (see figure \ref{fig6}) the element perturbed is put backwards, and as before the elements ahead increase the difference in phase in relation to the perturbed element; at each cycle, after being reset, the original $S{1}$ element receive an additional pulse coming from the perturbed element, increasing its phase. After some cycles this increase makes the  $S_{1}$ element to join the perturbed element, forming a new cluster $S_{2}$. The new configuration becomes $S_{3} \times S_{2}$, where $S_{3}$  is unstable. A second perturbation to the $S_{3}$ cluster moves the perturbed element to the cluster $S_{2}$. A last perturbation can either put the system in its initial configuration or split the cluster into two, depending on the position on the periodic orbit it is applied. The symmetry of this state is obviously not preserved for negative perturbations, and computer simulations indicate that more general perturbations bring an even more complicated switching dynamic due the large number of elements in the unstable cluster.

Considering only single positive perturbations, we can state a transition rule between two states when subject to a single positive perturbation,
\begin{equation}
 (a+\delta_{1},a,a,a,b) \rightarrow (b,a,a,a,a).
\end{equation}
Under these considerations, the resulting transition diagram is a fully connected one, and it is possible to jump from one equivalent permutation state to any other one applying only one perturbation.

\subsection{Clustered state $S_{2} \times S_{2} \times S_{1}$}
For the symmetry $S_{2} \times S_{2} \times S_{1}$, assuming $\delta_{3}<\delta_{4}$ and $\delta_{2},\delta_{5} > 0$, we have the following Jacobian matrix:
\begin{equation} \label{J-s2s2s1}
{\frac{\partial\bm{\delta}(n+1)}{\partial\bm{\delta}(n)} \bigg{|}}_{\bm{\delta}(n)=0} = \left( \begin{array}{cccc}
\alpha & 0 & 0 & 0 \\
\beta & 0 & 0 & 0 \\
\beta & 0 & 0 & 0 \\ 
\gamma & 0 & 0 & 0 \\
\end{array} \right),
\end{equation}
where $\alpha > \beta > \gamma > 0$ (see \ref{app-deriv}), cf tables \ref{pt-s2s2s1} and \ref{ut-s4s1}. This matrix has three zero eigenvalues and only one non-zero eigenvalue given by
\begin{flalign}
\lambda_{1}&= \alpha = H_{\epsilon}^{'}\left(\tau^{'}+H_{\epsilon}\left(\tau-\tau^{'}+H_{2\epsilon}(\tau^{'})\right)\right)\left[-1+H_{\epsilon}^{'}\left(\tau-\tau^{'}+H_{2\epsilon}(\tau^{'})\right)\left[1+H_{2\epsilon}^{'}(\tau^{'})\right]\right],
\end{flalign}
which is larger than one accordingly to Lemma \ref{lemma1}. The fact that there is only one eigenvalue and that it is larger than one not only shows that there is only one unstable cluster, but also that perturbations change only the difference in phase between the two elements on this cluster. As a result, any general perturbation can be mapped to a single oscillator perturbation.

Differently from the last two considered symmetries, we here have one single element, one stable $S_{2}$ clusters, and one symmetric unstable $S_{2}$ cluster. When perturbed, the initial unstable cluster $S_{2}$ splits into two, the additional pulse received now by the initial single $S_{1}$ element just after its reset makes it approaches the element that was put behind on the unstable $S_{2}$ cluster, forming a new stable $S_{2}$ cluster. This occurs because it is reset by supra-threshold pulses. Moreover the element ahead begins to be reset by pulses and stops increasing its phase, becoming stable and the original stable $S_{2}$ cluster after changing its phase, is not reset by pulses anymore, becoming unstable. The final state has the same symmetry and stability properties as the former state.

\begin{figure}[]
\begin{center}
\includegraphics[width=11.0 cm,clip,angle=0]{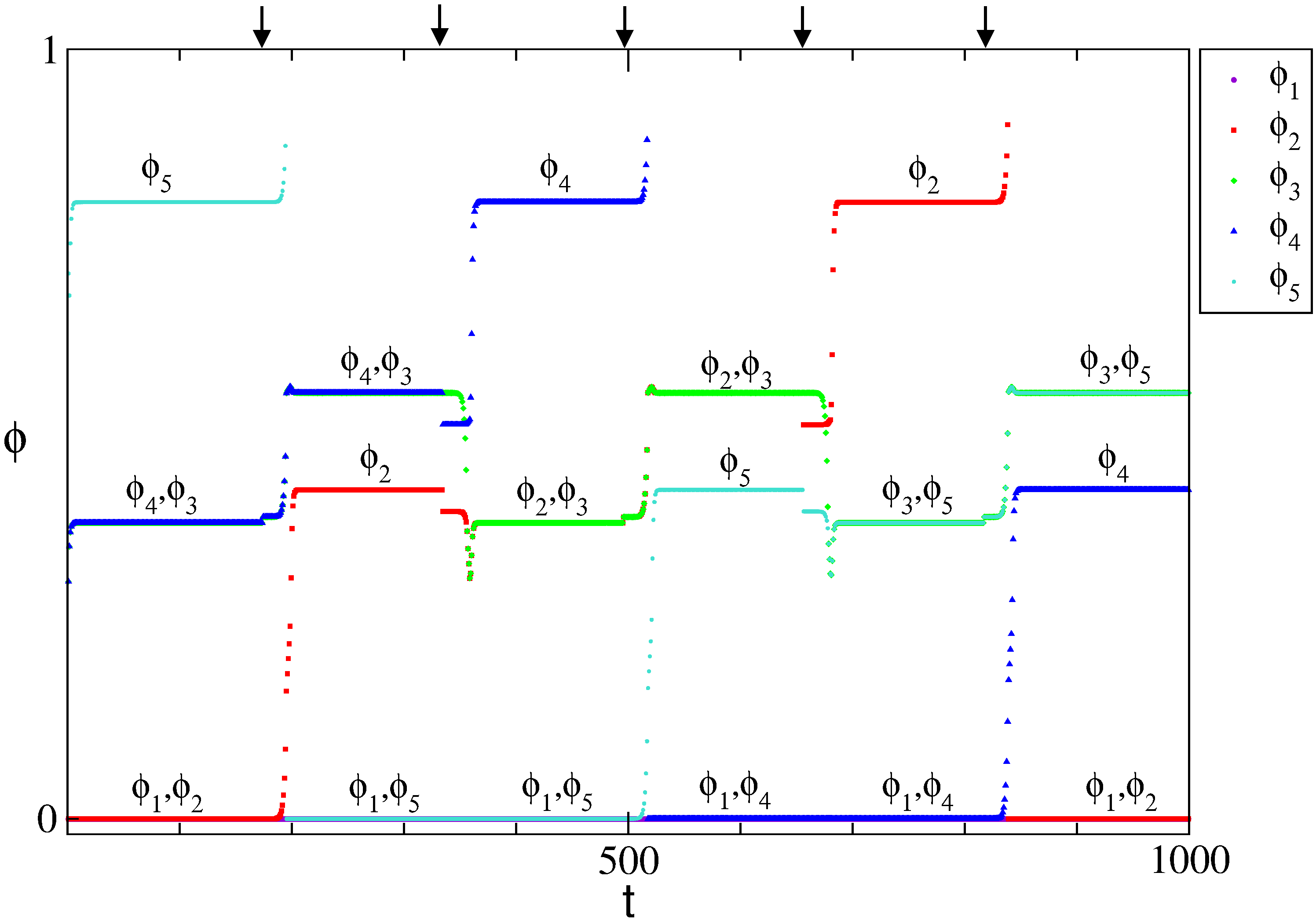}
\end{center}
\caption[1st-entry]{ \footnotesize Example of perturbation-induced switching in a $S_{2} \times S_{2} \times S_{1}$ state set. A sequence of five perturbations driving the system through different states with symmetry $S_{2} \times S_{2} \times S_{1}$ ($\tau=0.49$,$\epsilon=0.025$,$I=1.04$,$\gamma=1$). The phase of all oscillators are plotted at the moment when oscillator number one is reset, each color representing the phase of one oscillator. The apparent change of the phase differences among the clusters just after the perturbations depends on the cluster to which the reference oscillator belongs.  The symmetry of the unstable attractors is preserved. The sequence of states corresponding to the plateaus are $(a,a,b,b,c) \rightarrow (b,c,a,a,b)\rightarrow (a,b,b,c,a) \rightarrow (b,a,a,b,c) \rightarrow (a,c,b,a,b) \rightarrow (b,b,a,c,a).$}
\label{fig7}
\end{figure}

The preservation of the symmetry implies a closed transition diagram among all the possible $S_{2} \times S_{2} \times S_{1}$ states (see figure \ref{diagram-s2s2s1}). We state two simple equivalent switching rules. Considering first a positive representation we have
\begin{equation}\label{s2s2s1-rule1}
 (a,a+\delta_{2},b,b,c) \rightarrow (c,b,a,a,b)
\end{equation}
that can be rewritten for negative perturbations simply as
\begin{equation}\label{s2s2s1-rule2}
 (a-\delta_{1},a,b,b,c) \rightarrow (c,b,a,a,b).
\end{equation}

\begin{figure}[]
\begin{center}
\includegraphics [width=12.0cm,clip,angle=0]{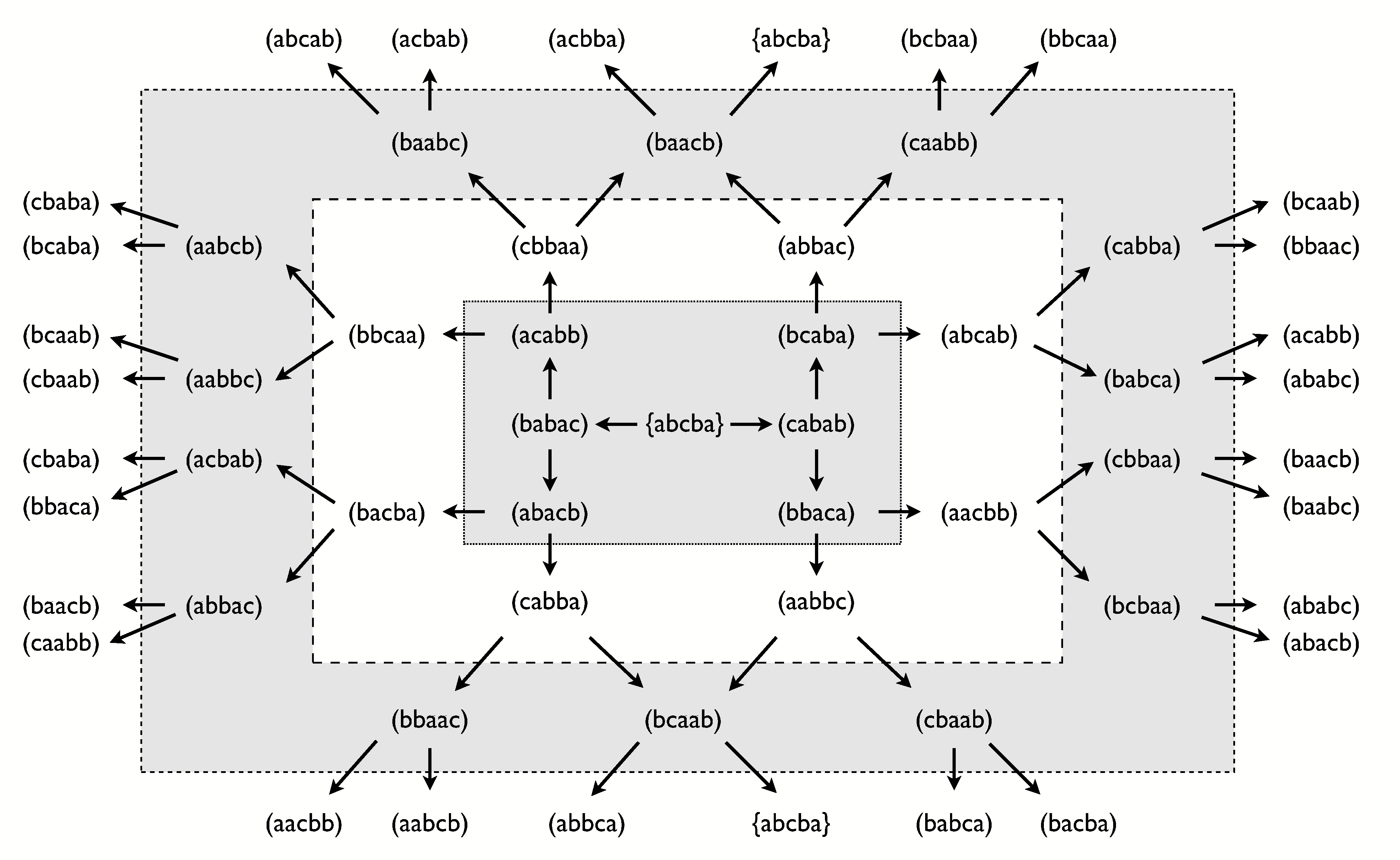}
\end{center}
\caption{Five steps state transition diagram for the symmetry $S_{2} \times S_{2} \times S_{1}$. This diagram shows all the possible 5 steps paths, sequence of states, beginning on the state \{a,b,c,b,a\}. Each arrow correspond to one of the two possible perturbations, subject to (\ref{s2s2s1-rule1}) and (\ref{s2s2s1-rule2}). It is necessary at least 5 perturbations to reach the initial state again.
}
\label{diagram-s2s2s1}
\end{figure}

We conclude that for this symmetry the unstable attractors are linked to form a heteroclinic network (figure \ref{diagram-s2s2s1}), characterized by (\ref{s2s2s1-rule1}) and (\ref{s2s2s1-rule2}), forming a closed set of saddle periodic orbits among which the systems switches in a controlled way upon small external perturbations. We remark that in the absence of noise this dynamic does not exhibit spontaneous transitions between nearby saddle states \cite{BES2008-2,CT2008} but instead displays convergence to unstable attractors. The free dynamic of each element evolves continuously up to reset; still, the collective dynamic of the entire network (network state) evolves continuously almost always, but interrupted by discrete jumps due to the infinitely fast phase response of the interaction. As (i) the transitions are fully controlled by external perturbations and thus predictable, and (ii) the symmetry is preserved when the network is subject to sufficiently small, general perturbations, arbitrarily small external noise would trigger a persistent switching dynamic in which the network states are constrained to the closed set of periodic orbits with initial $S_2\times S_2\times S_1$ symmetry. A numerical example of this spontaneous switching phenomenon has been reported before \cite{TWG2003} for a system of $N=100$ oscillators exhibiting $S_{21}\times S_{21} \times S_{21}\times S_{21} \times S_{16}$ symmetry.

\section{Discussion}
In the networks of pulse-coupled oscillators studied above, three sets of heteroclinically connected unstable attractors appear to have a well-defined symmetry that depends on the network parameters. Interestingly, for two state symmetries, the possible switching transitions markedly deviate from those in time-continuously coupled systems \cite{G2007}. Moreover, all attractors with the third symmetry $S_{2} \times S_{2} \times S_{1}$ form one closed heteroclinic network, where all possible transitions are predictable and depend on the precise direction of the perturbation. In fact, mapping an arbitrary small perturbation to a single oscillator perturbation, we derived a general set of transition rules, (\ref{s2s2s1-rule1}) and (\ref{s2s2s1-rule2}). This last feature guarantees that there are no changes of symmetry during the switching and precisely defines a transition diagram (figure \ref{diagram-s2s2s1}) that holds for all sufficiently small perturbations.

Thus, this work explicitly shows how nontrivial switching dynamics is induced and precisely controlled by perturbations in pulse-coupled systems. Our analysis shows that and how event sequences, collectively generated by the network, fully determine switching transitions in pulse-coupled oscillators. As a consequence, these results are not restricted to the IF model (used here for numerical simulations and illustrating purposes) but equally hold for different oscillator models with sub-threshold potential dynamics that are sufficiently close to the one considered here. Moreover, the phenomenon should still hold qualitatively for temporally extended responses as long as the post-synaptic response times are  short compared to the membrane time constant and inter-spike-interval times. Nevertheless, although we expect the same transition possibilities, the dynamics without noise will show heteroclinic switching sequences that depend on the initial condition and do not require external perturbations, cf \cite{CT2008}. As stability and instability of clusters reflect synchronizing and desynchronizing mechanisms \cite{KGT2009}, here realized by supra- and sub-threshold inputs, respectively, similar switching features also occur in networks of $N \geq 6$ pulse-coupled oscillators \cite{TWG2003}.

Since the systems studied here are pulse-coupled and of hybrid type, with smooth time evolution interrupted at discrete times of interactions, it is interesting to compare our results to those on systems of oscillators coupled continuously in time \cite{AB2004,AB2005}. The latter systems exhibit partially synchronized saddle states with the same symmetry $S_{2} \times S_{2} \times S_{1}$, where a persistent switching dynamics appears as one feature of the model and, when subject to asymmetric external currents, generates a wide variety of spatiotemporal patterns. Interestingly, the transitions rules given by (\ref{s2s2s1-rule1}) and (\ref{s2s2s1-rule2}) (and illustrated on figures \ref{fig7} and \ref{diagram-s2s2s1}) not only guarantee an equivalent persistent switching dynamic when subject to noise, but also imply as well the existence of the same set of spatiotemporal patterns when subject to asymmetric external currents. Thus, our model characterizes exactly this switching dynamic in a pulse-coupled neuronal framework, where the 
patterns can be described as distributed pulse-sequences (spike patterns). The importance of such a spiking representation becomes evident in particular when considering potential applications to neural coding and information processing \cite{AB2005}. For instance, studies on the olfactory system of insects \cite{ML2005,L2002} have shown that biological systems could use spatiotemporal spike patterns as part of their information processing. In particular our results agree with the interesting predictions of Hansel et al \cite{H1993}, Rabinovich et al \cite{R2001}, and Timme et al \cite{TWG2002} regarding the generation of spatiotemporal spike patterns based on a switching dynamics. In addition our work presents a neural system where the entire (long-time) switching dynamics follows from a fixed set of transition rules, a promising feature that may prove not only advantageous for computation in biological but also in artificial systems.

We remark that although the apparent equivalence between the dynamics of pulse-coupled oscillators and continuously coupled oscillators works for this specific symmetry, it does not holds as a general rule. The most pronounced counter-examples are systems with symmetry $S_{3} \times S_{2}$, which when smoothly coupled exhibit persistent switching dynamics, but when pulse-coupled, any continuous small noise required for the switching necessarily drives the system to a stable attractor, cf figure \ref{fig4}.

To understand how these switching properties may actually perform computational tasks, a complete analysis of the effect of asymmetric currents, driving pulses and asymmetric connections on the spike patterns is needed. The answer to these controlling factors could bring us important information about alternative mechanisms of neural computation, both biological and artificial.

\begin{acknowledgments}
We thank Frank van Bussel for critically checking the English in this manuscript. We thank the Federal Ministry of Education and Research (BMBF) Germany for support under grant number 01GQ0430 to the Bernstein Center for Computational Neuroscience (BCCN) G\"{o}ttingen.
\end{acknowledgments}

\appendix

\section{Partial derivatives}\label{app-deriv}
Here we present the analytical expressions for the Jacobian matrices presented on section \ref{sec-PDY}. Where we introduce a short notation $H_{y}(x) \rightarrow x_{y}$, for $x \in \left\{0,\tau,(\tau-\tau^{'}),(\tau-\tau^{'}+\tau_{y}^{'})\right\}$ and $y \in \{\epsilon,2\epsilon,3\epsilon\}$

\subsection{$S3\times S2$: non-zero elements for the Jacobian matrix (\ref{J-s3s2})}
\begin{align}
\alpha &=
\left[-1+\left[1+H_{\epsilon}^{'}(\tau)\right]H_{\epsilon}^{'}(\tau_{\epsilon})\right]H_{2\epsilon}^{'}(\tau+\tau_{2\epsilon})\\
\beta &= \left[-1+2H_{\epsilon}^{'}(\tau)\right]H_{\epsilon}^{'}(\tau_{\epsilon})H_{2\epsilon}^{'}(\tau+\tau_{2\epsilon})\\
\gamma &= \left[-1+H_{\epsilon}^{'}(0)\right]H_{2\epsilon}^{'}(\tau+0_{\epsilon})+H_{\epsilon}^{'}\left(\tau_{\epsilon}\right)H_{2\epsilon}^{'}\left(\tau+\tau_{2\epsilon}\right)\\
j_{21} &=-\left[1+\left[-2+H_{\epsilon}^{'}(\tau)\right]H_{\epsilon}^{'}(\tau_{\epsilon})\right]H_{2\epsilon}^{'}(\tau+\tau_{2\epsilon})\\
j_{32} &=j_{42}=\left[-1+H_{\epsilon}^{'}(\tau)\right]H_{\epsilon}^{'}(\tau_{\epsilon})H_{2\epsilon}^{'}(\tau+\tau_{2\epsilon})
\end{align}

\subsection{$S4\times S1$: non-zero elements for the Jacobian matrix (\ref{J-s4s1})} 
\begin{align}
\alpha &=\left[-1+\left[1+H_{\epsilon}^{'}(\tau)H_{\epsilon}^{'}\left(\tau_{\epsilon}\right) \right]H_{\epsilon}^{'}\left(\tau_{2\epsilon}\right)\right]H_{\epsilon}^{'}\left(\tau+\tau_{3\epsilon}\right)\\
\beta &=\left[-1+\left[1+H_{\epsilon}^{'}(\tau)\right]H_{\epsilon}^{'}\left(\tau_{\epsilon}\right) \right]H_{\epsilon}^{'}\left(\tau_{2\epsilon}\right)H_{\epsilon}^{'}\left(\tau+\tau_{3\epsilon}\right) \\
\gamma &=\left[-1+2H_{\epsilon}^{'}(\tau)\right]H_{\epsilon}^{'}\left(\tau_{\epsilon}\right)H_{\epsilon}^{'}\left(\tau_{2\epsilon}\right)H_{\epsilon}^{'}\left(\tau+\tau_{3\epsilon}\right) \\
j_{31} &=j_{21}=-\left[1+\left[-2+H_{\epsilon}^{'}(\tau_{\epsilon})\right]H_{\epsilon}^{'}(\tau_{2\epsilon})\right]H_{\epsilon}^{'}(\tau+\tau_{3\epsilon}) \\
j_{41} &=-\left[-1+H_{\epsilon}^{'}(0_{\epsilon})\right]H_{\epsilon}^{'}(0_{2\epsilon})+\left[-1+H_{\epsilon}^{'}(\tau_{\epsilon})\right]H_{\epsilon}^{'}(\tau+\tau_{\epsilon})
\end{align}
\begin{flalign}
j_{42} &=-\left[-1+H_{\epsilon}^{'}(0)\right]H_{\epsilon}^{'}(0_{\epsilon})H_{\epsilon}^{'}(0_{2\epsilon})+\left[-1+H_{\epsilon}^{'}(\tau_{\epsilon})\right]H_{\epsilon}^{'}(\tau_{2\epsilon})H_{\epsilon}^{'}(\tau+\tau_{3\epsilon}) \\
\theta &=-1+H_{\epsilon}^{'}(0)H_{\epsilon}^{'}(0_{\epsilon})H_{\epsilon}^{'}(0_{2\epsilon})+\left[1+\left[-1+H_{\epsilon}^{'}(\tau)\right]H_{\epsilon}^{'}(\tau_{\epsilon})H_{\epsilon}^{'}(\tau_{2\epsilon})\right]H_{\epsilon}^{'}(\tau+\tau_{3\epsilon})
\end{flalign}

\subsection{$S2 \times S2 \times S1$: non-zero elements for the Jacobian matrix (\ref{J-s2s2s1}).}
\begin{flalign}
\lambda_{1} &= \alpha = H_{\epsilon}^{'}\left(\tau^{'}+\left(\tau-\tau^{'}+\tau^{'}_{2\epsilon}\right)_{\epsilon}\right)\left[-1+H_{\epsilon}^{'}\left(\tau-\tau^{'}+\tau^{'}_{2\epsilon}\right)\left[1+H_{2\epsilon}^{'}(\tau^{'})\right]\right], \\
\beta &= H_{\epsilon}^{'}(0)H_{\epsilon}^{'}(\tau^{'}+0_{\epsilon})+H_{\epsilon}^{'}\left(\tau^{'}+(\tau-\tau^{'}+\tau_{2\epsilon})_{\epsilon}\right)\left[-1+H_{\epsilon}^{'}(\tau-\tau^{'}+\tau_{2\epsilon}^{'})\right]\\
 \gamma &= -\left[-1+H_{\epsilon}^{'}(\tau-\tau^{'})\right]H_{\epsilon}^{'}\left((\tau-\tau^{'})_{\epsilon}\right)+H_{\epsilon}^{'}\left(\tau^{'}+(\tau-\tau^{'}+\tau^{'}_{\epsilon})_{\epsilon}\right)\left[-1+H_{\epsilon}^{'}(\tau-\tau^{'}+\tau_{2\epsilon}^{'})\right]
\end{flalign}

\section{Return Maps}\label{app-rMaps}
Here we explain step by step the periodic orbit dynamic described by the unperturbed return maps that define the three main families of attractors, given in tables \ref{ut-s3s2-1}, \ref{ut-s3s2-2}, \ref{ut-s4s1} and \ref{ut-s2s2s1}. The event notation is the following: $s_{i}$ indicates that oscillator $i$ sent a pulse; $r_{i}$ indicates that pulses were received coming from the oscillators indicated by $i$. Capital letters indicate constants,  $H_{\epsilon}(\phi)$ is the transfer function presented on section 1, and $p_{i,j}$ is a short notation for the phase of oscillator $i$ at event number $j$.

A realization of the dynamics described by these tables are presented for specific parameter in tables \ref{dt-s3s2}, \ref{dt-s4s1} and \ref{dt-s2s2s1}.  As there are no approximation to the corresponding analytical condition tables, the specific values completely agree with the iterated simulation in figures \ref{fig3}, \ref{fig5} and \ref{fig7}.

\subsection*{Unperturbed $S_{3} \times S_{2}$ dynamic }
The initial condition is such that no pulse was sent before time zero. At time zero, the first event, oscillators 1, 2 and 3 fire ($s_{123}$); the second event is the reception of these signals a time $\tau$ later ($r_{123}$), these is an supra-threshold event to oscillators 4 and 5, which then send a signal ($s_{45}$) and are reset; the third event is the reception of pulses from 4 and 5 ($r_{45}$); and the last event is the reset of oscillators 1, 2 and 3 ($s_{123}$) by reaching the threshold. For any choice of the parameters resulting in $A=p_{4,2}+1-p_{1,2}$ (while preserving the event sequence), we have a period-one attractor, since the initial state is obtained after one pulse of each oscillator. A numerical example of such a structure is presented in table \ref{dt-s3s2}. From this map, we can conclude that the cluster $S_{2}$ is stable, since any small variation will be restored when its elements are reset together by the incoming pulse.

\begin{table}[ht]
\caption{Analytic table of condition for an unperturbed $S_{3}\times S_{2}$ dynamic, $S_{3}$ unstable.}
\label{ut-s3s2-1}
\footnotesize
\begin{itemize}[leftmargin=.0in]
\item[]\begin{tabular}{l  l  l  l  l }
\hline \hline
event & time &  $\phi_{1},\phi_{2},\phi_{3}$  &  $\phi_{4},\phi_{5}$  & event num. \\
\hline
$s_{1,2,3}$ & 0 & 0 & A & 0 \\
$r_{1,2,3}; s_{4,5}$ & $\tau$ & $H_{2\epsilon}\left(\tau\right) = p_{1,1}$ & $H_{3\epsilon}\left(A+\tau\right)>1 \rightarrow 0$ & 1 \\
$r_{4,5}$ & $2 \tau$ & $H_{2 \epsilon}\left(p_{1,1}+\tau\right) = p_{1,2}$ & $H_{\epsilon}\left(\tau\right)=p_{4,2}$ & 2 \\
$s_{1,2,3}$ & $2 \tau+1-p_{1,2}$ & $ 1 \rightarrow 0$ & $p_{4,2}+1-p_{1,2}$ & 3 \\
\hline \hline
\end{tabular}
\end{itemize}
\end{table}

\begin{table}[ht]
\caption{Analytic prediction of phase dynamic for parameters $\tau=0.31$, $\epsilon=0.025$, $I=1.04$ and $\gamma=1$, realizing a   $S_{3} \times S_{2}$ cycle.}
\label{dt-s3s2}
\footnotesize
\begin{itemize}[leftmargin=.0in]
\item[]\begin{tabular}{ c l  l  l  c }
\hline \hline
event & time &  $\phi_{1},\phi_{2},\phi_{3}$  &  $\phi_{4},\phi_{5}$  & event num. \\
\hline
$s_{1,2,3}$ & 0.000000 & 0.000000 & 0.501612 & 0 \\
$r_{1,2,3}; s_{4,5}$ & 0.310000 & 0.353450 & 0.000000 & 1 \\
$r_{4,5}$ & 0.620000 & 0.829344 & 0.330956 & 2 \\
$s_{1,2,3}$ & 0.790655 & 0.000000 & 0.501612 & 3 \\
\hline\hline
\end{tabular}
\end{itemize}
\end{table}

\subsection*{Unperturbed $S_{3}\times S_{2}$ dynamic, $S_{2}$ unstable}
This map describe the partner orbit of table \ref{ut-s3s2-1}, once they appear for the same range of the parameter, but for different initial conditions. The initial condition here is that pulses from oscillators 4 and 5 are received ($r_{(4,5)}$) exactly at time 0, forcing oscillators 1, 2 and 3 to fire ($s_{1,2,3}$), what define the first event; the second event is the reception of these pulses ($r_{1,2,3}$) $\tau$ time later; the third event is the reset of oscillators 4 and 5 upon reaching the threshold ($s_{4,5}$), and consequently the generation of two new pulses; the last event is the reception of these pulses ($r_{4,5}$), which causes oscillators 1, 2 and 3 to generate one pulse ($s_{1,2,3}$). In this case the $S_{3}$ cluster is stable, since any small variation on the phase of its components will disappear in the next cycle when all are reset together by incoming pulses \cite{AT2005}.

\begin{table}[ht]
\caption{Analytic table of condition for an unperturbed $S_{3} \times S_{2}$ dynamic, $S_{2}$ unstable.}
\label{ut-s3s2-2}
\centering
\footnotesize
\begin{itemize}[leftmargin=.0in]
\item[]\begin{tabular}{l  l  l  l  l }
\hline\hline
 event & time &  $\phi_{1},\phi_{2},\phi_{3}$  &  $\phi_{4},\phi_{5}$  & event num. \\
\hline
$r_{(4,5)};s_{1,2,3}$ & 0 & 0 & C & 0 \\

$r_{1,2,3}$ & $\tau$ & $H_{2\epsilon}\left(\tau\right) = p_{1,1}$ & $H_{3\epsilon}\left(C+\tau\right)=p_{4,1}$ & 1 \\

$s_{4,5}$ & $1-p_{4,1}$ & $1+p_{1,1}-p_{4,1} = p_{1,2}$ & $1\rightarrow 0$ & 2 \\

$r_{4,5};s_{1,2,3}$ & $\tau+1-p_{4,1}$ & $H_{2\epsilon}\left(p_{1,2}+\tau\right) > 1 \rightarrow 0$ & $H_{\epsilon}\left(\tau\right)$ & 3 \\
\hline\hline
\end{tabular}
\end{itemize}
\end{table}

\subsection*{Unperturbed $S_{4} \times S_{1}$ dynamic }
This map describes another period-one attractor, where no pulse was sent before time 0. The first event is the signal sent by oscillators 1, 2, 3 and 4 ($s_{1,2,3,4}$); the second event is the reception of these pulses after $\tau$ time units ($r_{1,2,3,4}$), which makes oscillator 5 to generate one pulse due to a supra-threshold input ($s_{5}$); the third event is the reception of this pulse at time $2\tau$ ($r_{5}$); the last event is the pulse generation from oscillators 1, 2, 3, and 4 upon reaching the threshold. If the event sequence is preserved, the condition $B=\tau+1-p_{1,2}$ follows from the periodicity of the orbit. A numerical example of this orbit structure is presented in table \ref{dt-s4s1}.
\begin{table}[ht]
\caption{Analytic table of condition for an unperturbed $S_{4} \times S_{1}$ dynamic.}
\label{ut-s4s1}
\centering
\footnotesize
\begin{itemize}[leftmargin=.0in]
\item[]\begin{tabular}{l  l  l  l  l }
\hline \hline
 event & time &  $\phi_{1},\phi_{2},\phi_{3},\phi_{4}$  &  $\phi_{5}$  & event num. \\
\hline
$s_{1,2,3,4}$ & 0 & 0 & B & 0 \\

$r_{1,2,3,4}; s_{5}$ & $\tau$ & $H_{3\epsilon}\left(\tau\right) = p_{1,1}$ & $H_{4\epsilon}\left(B+\tau\right)>1 \rightarrow 0$ & 1 \\
$r_{5}$ & $2 \tau$ & $H_{ \epsilon}\left(p_{1,1}+\tau\right) = p_{1,2}$ & $\tau$ & 2 \\
$s_{1,2,3,4}$ & $2 \tau+1-p_{1,2}$ & $ 1 \rightarrow 0$ & $\tau+1-p_{1,2}$ & 3 \\
\hline\hline
\end{tabular}
\end{itemize}
\end{table}

\begin{table}[ht]
\caption{Analytic prediction of phase dynamic for parameters $\tau=0.27$, $\epsilon15$, $I=1.1$ and $\gamma=1$, realizing a   $S_{4} \times S_{1}$ cycle.}
\label{dt-s4s1}
\centering
\footnotesize
\begin{itemize}[leftmargin=.0in]
\item[]\begin{tabular}{ c l  l  l  c }
\hline\hline
 event & time &  $\phi_{1},\phi_{2},\phi_{3},\phi_{4}$  &  $\phi_{5}$  & event num. \\
\hline
$s_{1,2,3,4}$ & 0.000000 & 0.000000 & 0.672908 & 0 \\
$r_{1,2,3,4}; s_{5}$ & 0.270000 & 0.303940 & 0.000000 & 1 \\
$r_{5}$ & 0.540000 & 0.597091 & 0.270000 & 2 \\
$s_{1,2,3,4}$ & 0.942909 & 0.000000 & 0.672908 & 3 \\
\hline\hline
\end{tabular}
\end{itemize}
\end{table}

\subsection*{Unperturbed $S_{2} \times S_{2} \times S_{1}$ dynamic}
For this map, the initial conditions are such that pulses from oscillators 3 and 4 will be received at time $\tau^{'}$ after time 0. The first event is the pulse generation from oscillators 1 and 2 upon reaching the threshold ($s_{1,2}$); the second is the reception of pulses from oscillators 3 and 4 ($r_{(3,4)}$) at time $\tau^{'}$ and the pulse generation from oscillator 5 ($s_{5}$) caused by this supra-threshold input; the third event is the receive of the pulses from oscillators 1 and 2 ($r_{1,2}$) at time $\tau$ that forces oscillators 3 and 4 to elicit a pulse ($s_{3,4}$); the forth event is the reception of the pulse coming from oscillator 5 ($r_{5}$); the last event is the pulse generation from oscillators 1 and 2 ($s_{1,2}$) upon reaching the threshold. This map implies three periodic conditions to describe a period-one attractor: $D=p_{3,3} +1 - p_{1,3}$, $E = p_{5,3}+1-p_{1,3}$, and $\tau=2\tau^{'}+1-p_{1,3}$. A example of this structure is presented in table \ref{dt-s2s2s1}.

\begin{table}[ht]
\caption{Analytic table of condition for an unperturbed $S_{2} \times S_{2} \times S_{1}$ dynamic.}
\label{ut-s2s2s1}
\centering
\scriptsize
\begin{itemize}[leftmargin=.0in]
\item[]\begin{tabular}{ p{1cm} p{1.3cm} p{3.4cm} p{3.5cm} p{3cm} p{0.5cm}}
\hline\hline
 event & time &  $\phi_{1},\phi_{2}$ & $\phi_{3},\phi_{4}$  &  $\phi_{5}$  & event num. \\
\hline
$s_{1,2}$ & 0 & 0 & D & E & 0 \\
$r_{(3,4)}; s_{5}$ & $\tau^{'}$ & $H_{2\epsilon}\left(\tau^{ }\right) = p_{1,1}$ & $H_{\epsilon}\left(D+\tau^{'}\right)=p_{3,1}$ & $H_{2\epsilon}\left(E+\tau^{'}\right)$ $>1 \rightarrow 0$ & 1 \\
$r_{1,2};s_{3,4}$ & $\tau$ & $H_{ \epsilon}\left(p_{1,1}+\tau-\tau^{'}\right) = p_{1,2}$ & $H_{ 2\epsilon}\left(p_{3,1}+\tau-\tau^{'}\right)$ $>1 \rightarrow 0$ & $H_{2\epsilon}\left(\tau -\tau^{'}\right)$ & 2 \\
$r_{5}$ & $\tau+\tau^{'}$ & $H_{\epsilon}\left(p_{1,2}+\tau^{'}\right)=p_{1,3}$ & $H_{\epsilon}\left(p_{3,2}+\tau^{'}\right)=p_{3,3}$ & $p_{5,2}+\tau^{'}=p_{5,3}$ & 3 \\
$s_{1,2}$ & $\tau+\tau^{'}+1-p_{1,3}$ & $1 \rightarrow 0$ & $p_{3,3} +1 - p_{1,3} $ & $p_{5,3}+1-p_{1,3}$ & 4 \\
\hline\hline
\end{tabular}
\end{itemize}
\end{table}

\begin{table}[ht]
\caption{Analytic prediction of phase dynamic for parameters $\tau=0.49$, $\epsilon=0.025$, $I=1.04$ and $\gamma=1$, realizing a   $S_{2} \times S_{2} \times S_{1}$ cycle.}
\label{dt-s2s2s1}
\centering
\footnotesize
\begin{itemize}[leftmargin=.0in]
\item[]\begin{tabular}{ c l  l  l  l  c }
\hline\hline
  event & time &  $\phi_{1},\phi_{2}$ & $\phi_{3},\phi_{4}$  &  $\phi_{5}$  & event num. \\
\hline
$s_{1,2}$ & 0.000000 & 0.000000 & 0.381978 & 0.795680 & 0 \\
$r_{(3,4)}; s_{5}$ & 0.119095 & 0.141656 & 0.541358 & 0.000000 & 1 \\
$r_{1,2};s_{3,4}$ & 0.490000 & 0.554491 & 0.000000 & 0.424775 & 2 \\
$r_{5}$ & 0.609095 & 0.748191 & 0.130168 & 0.543870 & 3 \\
$s_{1,2}$ & 0.860904 & 0.000000 & 0.381978 & 0.795680 & 4 \\
\hline\hline
\end{tabular}
\end{itemize}
\end{table}

\section{Perturbed dynamic, return maps.}
In this appendix we present three tables that show the changes to the dynamics described in tables \ref{ut-s3s2-1}, \ref{ut-s4s1} and \ref{ut-s2s2s1} due an incremental perturbation $\bm{\delta}=\left(0,\delta_{2},\delta_{3},\delta_{4},\delta_{5}\right)$, where $0<\delta_{2}<\delta_{3}<\delta_{4}<\delta_{5}\ll 1$. In all cases without lose of generality oscillator 1 was taken as the referential phase to define the new cycle, since this doesn't affect the dynamics itself but only the point of reference. The notation is the same as in \ref{app-rMaps}.

\begin{sidewaystable}[]
\caption{Perturbed $S_{3} \times S_{2}$ dynamic.}
\label{pt-s3s2}
\centering
\scriptsize
\begin{tabular}{p{1cm} p{1.3cm} p{3cm} p{3cm} p{3cm} p{3cm} p{3cm} p{0.5cm}}	
\hline\hline
 event & time &  $\phi_{1}$ &  $\phi_{2}$ & $\phi_{3}$ & $\phi_{4}$  &  $\phi_{5}$  & event num. \\
\hline
&&&&&&&\\
$s_{1,(2,3)}$ & 0 & 0 & $\delta_{2}$ & $\delta_{3}$ & $A+\delta_{4} $& $A+\delta_{5}$ & 0 \\
&&&&&&&\\
&&&&&&&\\
$r_{3}$ & $\tau-\delta_{3}$ & $H_{\epsilon}\left(\tau-\delta_{3} \right) = p_{1,1a}$ & $H_{\epsilon}\left(\delta_{2} +\tau-\delta_{3} \right)$  $=p_{2,1a}$ & $\tau$ & $H_{\epsilon}\left(A+\delta_{4}+\tau-\delta_{3} \right)$ $=p_{4,1a}$ & $H_{\epsilon}\left(A+\delta_{5}+\tau-\delta_{3} \right)$ $=p_{5,1a}$ & 1a \\
&&&&&&&\\
&&&&&&&\\
$r_{2};s_{4,5}$ & $\tau-\delta_{2}$ & $H_{ \epsilon}\left(p_{1,1a}+\delta_{3}-\delta_{2}\right)$ $ = p_{1,1b}$ & $p_{2,1a}+\delta_{3}-\delta_{2}=p_{2,1b}$ & $H_{\epsilon}\left(\tau +\delta_{3}-\delta_{2}\right)$ $=p_{3,1b}$ & $H_{\epsilon}\left(p_{4,1a}+\delta_{3}-\delta_{2}\right)$ $>1 \rightarrow 0$ & $H_{\epsilon}\left(p_{5,1a}+\delta_{3}-\delta_{2}\right)$ $>1 \rightarrow 0$ & 1b \\
&&&&&&&\\
&&&&&&&\\
$r_{1}$ & $\tau$ & $p_{1,1b}+\delta_{2}=p_{1,1}$ &$H_{\epsilon}\left(p_{2,1b}+\delta_{2}\right)=p_{2,1}$ & $H_{\epsilon}\left(p_{3,1b}+\delta_{2} \right)=p_{3,1}$ & $H_{\epsilon}\left(\delta_{2}\right)=p_{4,1}$ & $p_{4,1}$ & 1  \\
&&&&&&&\\
&&&&&&&\\
$r_{4,5}$ & $2\tau-\delta_{2}$ & $H_{2\epsilon}\left(p_{1,1}+\tau-\delta_{2}\right)$ $=p_{1,2}$ & $H_{2\epsilon}\left(p_{2,1}+\tau-\delta_{2}\right)$ $=p_{2,2}$ & $H_{2\epsilon}\left(p_{3,1}+\tau-\delta_{2}\right)$ $=p_{3,2}$ & $H_{2\epsilon}\left(p_{4,1}+\tau-\delta_{2}\right)$ $=p_{4,2}$ & $p_{4,2}$  & 2 \\
&&&&&&&\\
&&&&&&&\\
$s_{3}$ & $2\tau-\delta_{2}$ $+1-p_{3,2}$ & $p_{1,2}+1-p_{3,2}$ & $p_{2,2}+1-p_{3,2}$ & $1 \rightarrow 0$ & $p_{4,2}+1-p_{3,2}$ & $p_{4,2}+1-p_{3,2}$ & 3a \\
&&&&&&&\\
&&&&&&&\\
$s_{2}$ & $2\tau-\delta_{2}$ $+1-p_{2,2}$ & $p_{1,2}+1-p_{2,2}$ & $1 \rightarrow 0$ & $p_{3,2}-p_{2,2}$& $p_{4,2}+1-p_{2,2}$ & $p_{4,2}+1-p_{2,2}$ & 3b \\
&&&&&&&\\
&&&&&&&\\
$s_{1}$& $2\tau-\delta_{2}$ $+1-p_{1,2}$ & $1 \rightarrow 0$ & $p_{2,2}-p_{1,2}$& $p_{3,2}-p_{1,2}$ & $p_{4,2}+1-p_{1,2}$ & $p_{4,2}+1-p_{1,2}$ & 3 \\
\hline\hline
\end{tabular}
\end{sidewaystable}

\begin{sidewaystable}[]
\caption{Perturbed $S_{4} \times S_{1}$ dynamic.}
\label{pt-s4s1}
\centering
\scriptsize
\begin{tabular}{p{1cm} p{1.3cm} p{3cm} p{3cm} p{3cm} p{3cm} p{3cm} p{0.5cm}}
\hline\hline
 event & time &  $\phi_{1}$ &  $\phi_{2}$ & $\phi_{3}$ & $\phi_{4}$  &  $\phi_{5}$  & event num. \\
\hline
&&&&&&&\\
$s_{1,(2,3,4)}$ & 0 & 0 & $\delta_{2}$ & $\delta_{3}$ & $\delta_{4} $& $B+\delta_{5}$ & 0 \\
&&&&&&&\\
&&&&&&&\\
$r_{4};s_{5}$ & $\tau-\delta_{4}$ & $H_{\epsilon}\left(\tau-\delta_{4} \right) = p_{1,1a}$ & $H_{\epsilon}\left(\tau+\delta_{2}-\delta_{4} \right)=p_{2,1a}$ & $H_{\epsilon}\left(\tau+\delta_{3}-\delta_{4} \right)=p_{3,1a}$& $\tau$ & $H_{\epsilon}\left(B+\delta_{5}+\tau-\delta_{4} \right)>1 \rightarrow 0$  & 1a \\
&&&&&&&\\
&&&&&&&\\
$r_{3}$ & $\tau-\delta_{3}$ & $H_{ \epsilon}\left(p_{1,1a}+\delta_{4}-\delta_{3}\right) = p_{1,1b}$ & $H_{ \epsilon}\left(p_{2,1a}+\delta_{4}-\delta_{3}\right) = p_{2,1b}$ & $p_{3,1a}+\delta_{4}-\delta_{3}=p_{3,1b}$ & $H_{\epsilon}\left(\tau +\delta_{4}-\delta_{3}\right)=p_{4,1b}$ & $H_{\epsilon}\left(\delta_{4}-\delta_{3}\right)=p_{5,1b}$ & 1b \\
&&&&&&&\\
&&&&&&&\\
$r_{2}$ & $\tau-\delta_{2}$ & $H_{ \epsilon}\left(p_{1,1b}+\delta_{3}-\delta_{2}\right) = p_{1,1c}$ &$p_{2,1b}+\delta_{3}-\delta_{2}=p_{2,1c}$ & $H_{\epsilon}\left(p_{3,1b}+\delta_{3}-\delta_{2} \right)=p_{3,1c}$ & $H_{\epsilon}\left(p_{4,1b}+\delta_{3}-\delta_{2}\right)=p_{4,1c}$ & $H_{\epsilon}\left(p_{5,1b}+\delta_{3}-\delta_{2}\right)=p_{5,1c}$ & 1c  \\
&&&&&&&\\
&&&&&&&\\
$r_{1}$ & $\tau$ & $p_{1,1c}+\delta_{2}=p_{1,1}$ & $H_{\epsilon}\left(p_{2,1c}+\delta_{2} \right)=p_{2,1}$ & $H_{\epsilon}\left(p_{3,1c}+\delta_{2} \right)=p_{3,1}$ & $H_{\epsilon}\left(p_{4,1c}+\delta_{2} \right)=p_{4,1}$ & $H_{\epsilon}\left(p_{5,1c}+\delta_{2}\right)=p_{5,1}$ & 1  \\
&&&&&&&\\
&&&&&&&\\
$r_{5}$ & $2\tau-\delta_{4}$ & $H_{\epsilon}\left(p_{1,1}+\tau-\delta_{4}\right)=p_{1,2}$ & $H_{\epsilon}\left(p_{2,1}+\tau-\delta_{4}\right)=p_{2,2}$ & $H_{\epsilon}\left(p_{3,1}+\tau-\delta_{4}\right)=p_{3,2}$ & $H_{\epsilon}\left(p_{4,1}+\tau-\delta_{4}\right)=p_{4,2}$ & $p_{5,1}+\tau-\delta_{4}=p_{5,2}$  & 2 \\
&&&&&&&\\
&&&&&&&\\
$s_{4}$ & $2\tau-\delta_{4}+1-p_{4,2}$ & $p_{1,2}+1-p_{4,2}$ & $p_{2,2}+1-p_{4,2}$ & $p_{3,2}+1-p_{4,2}$ & $1 \rightarrow 0$ & $p_{5,2}+1-p_{4,2}$ & 3a \\
&&&&&&&\\
&&&&&&&\\
$s_{3}$ & $2\tau-\delta_{4}+1-p_{3,2}$ & $p_{1,2}+1-p_{3,2}$ & $p_{2,2}+1-p_{3,2}$ & $1 \rightarrow 0$ & $p_{4,2}-p_{3,2}$ & $p_{5,1}+\tau-\delta_{4}+1-p_{3,2}=p_{5,3b}$ & 3b \\
&&&&&&&\\
&&&&&&&\\
$s_{2}$ & $2\tau-\delta_{4}+1-p_{2,2}$ & $p_{1,2}+1-p_{2,2}$ & $1 \rightarrow 0$ & $p_{3,2}-p_{2,2}=p_{3,3c}$& $p_{4,2}-p_{2,2}$ & $p_{5,3b}+p_{3,2}-p_{2,2}=p_{5,3c}$ & 3c \\
&&&&&&&\\
&&&&&&&\\
$s_{1}$& $2\tau-\delta_{4}+1-p_{1,2}$ & $1 \rightarrow 0$ & $p_{2,2}-p_{1,2}$& $p_{3,2}-p_{1,2}$ & $p_{4,2}-p_{1,2}$ & $p_{5,3c}+p_{2,2}-p_{1,2}$ & 3 \\
\hline\hline
\end{tabular}
\end{sidewaystable}

\begin{sidewaystable}[]
\caption{Perturbed $S_{2} \times S_{2} \times S_{1}$ dynamic.}
\label{pt-s2s2s1}
\centering
\scriptsize
\begin{tabular}{p{1cm} p{1cm} p{3cm} p{0.1cm} p{3cm} p{3cm} p{0.1cm} p{3cm} p{0.1cm} p{2.7cm} p{0.5cm}}
\hline\hline
 event & time &  $\phi_{1}$ &&  $\phi_{2}$ & $\phi_{3}$ && $\phi_{4}$  &&  $\phi_{5}$ & event num. \\
\hline
&&&&&&&&&&\\
$(s_{3,4})$; $s_{1,(2))}$ & 0 & 0 && $\delta_{2}$ & $D+\delta_{3}$ && $D+\delta_{4} $&& $E+\delta_{5}$ & 0 \\
&&&&&&&&&&\\
&&&&&&&&&&\\
$r_{3,4};s_{5}$ & $\tau^{'}$ & $H_{2\epsilon}\left(\tau^{'} \right) = p_{1,1}$ && $H_{2\epsilon}\left(\tau^{'}+\delta_{2} \right)=p_{2,1}$ & $H_{\epsilon}\left(D+\tau^{'}+\delta_{3}\right)=p_{3,1}$&& $H_{\epsilon}\left(D+\tau^{'}+\delta_{4}\right)=p_{4,1}$ && $H_{2\epsilon}\left(E+\tau^{'}+\delta_{5} \right)>1 \rightarrow 0$  & 1 \\
&&&&&&&&&&\\
&&&&&&&&&&\\
$r_{2};s_{3,4}$& $\tau-\delta_{2}$ & $H_{ \epsilon}\left(p_{1,1}+\tau-\tau^{'}-\delta_{2}\right) = p_{1,2a}$ &&$p_{2,1}+\tau-\tau^{'}-\delta_{2}=p_{2,2a}$ & $H_{\epsilon}\left(p_{3,1}+\tau-\tau^{'}-\delta_{2} \right)>1 \rightarrow 0$ && $H_{\epsilon}\left(p_{4,1}+\tau-\tau^{'}-\delta_{2} \right)>1 \rightarrow 0$&& $H_{\epsilon}\left(\tau-\tau^{'}-\delta_{2}\right)=p_{5,2a}$ & 2a  \\
&&&&&&&&&&\\
&&&&&&&&&&\\
$r_{1}$ & $\tau$ & $p_{1,2a}+\delta_{2}=p_{1,2}$ && $H_{\epsilon}\left(p_{2,2a}+\delta_{2} \right)=p_{2,2}$ & $H_{\epsilon}\left(\delta_{2} \right)=p_{3,2}$ && 
$H_{\epsilon}\left(\delta_{2} \right)=p_{4,2}$ && $H_{\epsilon}\left(p_{5,2a}+\delta_{2}\right)=p_{5,2}$ & 2  \\
&&&&&&&&&&\\
&&&&&&&&&&\\
$r_{5}$ & $\tau-\tau^{'}$ & $H_{\epsilon}\left(p_{1,2}+\tau^{'}\right)=p_{1,3}$ && $H_{\epsilon}\left(p_{2,2}+\tau^{'}\right)=p_{2,3}$ & $H_{\epsilon}\left(p_{3,2}+\tau^{'}\right)=p_{3,3}$ && $p_{3,3}$ && $p_{5,2}+\tau^{'}$  & 3 \\
&&&&&&&&&&\\
&&&&&&&&&&\\
$s_{2}$ & $\tau+\tau^{'}+1-p_{2,3}$ & $p_{1,3}+1-p_{2,3}$ && $1 \rightarrow 0$ & $p_{3,3}+1-p_{2,3}=p_{3,4a}$&& $p_{3,4}$ && $p_{5,2}+\tau^{'}+1-p_{2,3}$ & 4a \\
&&&&&&&&&&\\
&&&&&&&&&&\\
$s_{1}$&  $\tau+\tau^{'}+1-p_{1,3}$  & $1 \rightarrow 0$ && $p_{2,3}-p_{1,3}$& $p_{3,3}+1-p_{1,3}$ && $p_{3,3}+1-p_{1,3}$ && $p_{5,2}+\tau^{'}+1-p_{1,3}$ & 4 \\
\hline\hline
\end{tabular}
\end{sidewaystable}

\clearpage

\bibliographystyle{unsrt}
\bibliography{ControlledPerturbation2009}

\end{document}